\pdfoutput=1
\documentclass[conference]{IEEEtran}

\usepackage[utf8]{inputenc} % Uses the utf8 input encoding
\usepackage[english]{babel}% Trennregeln
\usepackage[T1]{fontenc}% wichtig für Trennung von Wörtern mit Umlauten
\usepackage{csquotes}
\usepackage{makecell}
\usepackage{cancel}
\usepackage{amsfonts}
\usepackage{graphicx}
\usepackage{subcaption}
\usepackage{textcomp}
\usepackage{float}

\usepackage{hyperref}

\usepackage{transparent}

\usepackage{enumitem} 
\usepackage{parskip} 
\usepackage{amsmath}
\usepackage{latexsym}
\usepackage{mathtools}
\usepackage{physics}
\usepackage{amssymb}
\usepackage{comment}
\usepackage{arydshln}

\usepackage{usebib}
\usepackage{comment}
\usepackage{xpatch}
\usepackage{algorithm,algcompatible}

\usepackage{qcircuit}

\usepackage{biblatex}
\hypersetup{colorlinks,%
citecolor=black,%
filecolor=black,%
linkcolor=black,%
urlcolor=black
}
\usepackage{tabularx}
\newcolumntype{P}[1]{>{\hspace{0pt}}p{#1}}
\newcolumntype{K}[1]{>{\hspace{0pt}}X{#1}}
\newcolumntype{L}[1]{>{\raggedright\let\newline\\\arraybackslash\hspace{0pt}}m{#1}}

\newcolumntype{R}[1]{>{\raggedleft\let\newline\\\arraybackslash\hspace{0pt}}m{#1}}
\newcolumntype{C}{>{\centering\arraybackslash}X}

\addto\extrasenglish{}
\addto\extrasenglish{}
\addto\extrasenglish{}

\usepackage{titlesec}
\titleformat{\paragraph}[hang]{\normalfont\normalsize}{\theparagraph}{1em}{}
\titlespacing*{\subsubsection}{0pt}{1.5ex plus 1ex minus .2ex}{0.75em}
\titlespacing*{\paragraph}{0pt}{1.5ex plus 0.75ex minus .2ex}{1em}

\usepackage{tikz,booktabs,array}
\usetikzlibrary{positioning}
\makeatletter
\newcommand{\markZwicky}[1][]{\pgfutil@ifnextchar({\mark@Zwicky{#1}}{\mark@Zwicky{#1}()}}
\def\mark@Zwicky#1(#2)#3{%
   \tikz[every Zwicky picture,#1]{%
     \node[every Zwicky node,draw=none,inner sep=+\z@,outer sep=+\z@] {#3};
     \def\tikz@Mark@name{#2}%
     \ifx\tikz@Mark@name\pgfutil@empty\else
       \tikzset{every Zwicky node/.append style={name={#2}}}%
     \fi
     \node[every Zwicky node,overlay] {\phantom{#3}};
   }%
}
\newcommand{\tikzZwicky}[1][]{%
  \def\tikz@Zwicky@args{#1}%
  \let\tikz@Zwicky@list\pgfutil@gobble
  \let\tikz@Zwicky@first\pgfutil@empty
  \pgfutil@ifnextchar(\tikz@Zwicky@collect\tikz@Zwicky@finish
}
\def\tikz@Zwicky@collect(#1){%
  \ifx\tikz@Zwicky@first\pgfutil@empty
    \edef\tikz@Zwicky@first{#1}%
  \else
    \edef\tikz@Zwicky@list{\tikz@Zwicky@list,#1}%
  \fi
  \pgfutil@ifnextchar(\tikz@Zwicky@collect\tikz@Zwicky@finish
}
\def\tikz@Zwicky@finish{%
  \tikz[remember picture,overlay]
    \draw[every Zwicky connector,/expanded=\tikz@Zwicky@args]
      (\tikz@Zwicky@first) [/expanded={@Zwicky@list/.list={\tikz@Zwicky@list}}] [every Zwicky connect finish/.try];
}
\pgfkeys{/expanded/.code={\edef\pgfkeys@temp{{#1}}\expandafter\pgfkeysalso\pgfkeys@temp}}
\makeatother
\tikzset{
  @Zwicky@list/.style={insert path={to[every Zwicky connector how/.try] (#1)}},
  every Zwicky picture/.style={
    baseline,
    remember picture,
  },
  every Zwicky node/.style={
    remember picture,
    anchor=base,
    inner sep=+2pt
  },
  every Zwicky connector/.style={
    ultra thick,
    red!80!black,
    draw opacity=0.2,
    line cap=round,
    line join=round
  }
}

\makeatletter
\def\@IEEEsectpunct{.\ \,}
\def\paragraph{\@startsection{paragraph}{4}{\z@}{1.5ex plus 1.5ex minus 0.5ex}%
{0ex}{\itshape}}
\makeatother

\usepackage{amsthm}
\newtheorem{theorem}{Theorem}[section]

\DeclarePairedDelimiter{\ceil}{\lceil}{\rceil}

\bibliography{references}
\begin{document}

\pagestyle{plain}
\pagenumbering{arabic}

\begin{center}
    \LARGE \textbf{Grover Search for Portfolio Selection}
\end{center}

\begin{center}
A. Ege Yilmaz\footnote{Hochschule Luzern, Institut f\"ur Finanzdienstleistungen Zug IFZ, Suurstoffi~1, 6343 Rotkreuz, Email: <\href{mailto:ahmetege.yilmaz@hslu.ch}{ahmetege.yilmaz@hslu.ch}>}, Stefan Stettler\footnote{Abraxas Informatik AG, The Circle 68, 8058 Zürich, Phone: +41 58 660 16 85, Email: <\href{mailto:stefan.stettler@inventx.ch}{stefan.stettler@abraxas.ch}>}, Thomas Ankenbrand\footnote{Hochschule Luzern, Institut f\"ur Finanzdienstleistungen Zug IFZ, Suurstoffi~1, 6343 Rotkreuz, Phone: +41 41 757 67 23, Email: <\href{mailto:thomas.ankenbrand@hslu.ch}{thomas.ankenbrand@hslu.ch}>} and Urs Rhyner\footnote{Inventx AG, The Circle 37, 8058 Zürich, Phone: +41 81 287 19 79, Email: <\href{mailto:urs.rhyner@inventx.ch}{urs.rhyner@inventx.ch}>}
\bigskip
\noindent
\\
\bigskip

\today

\end{center}

\begin{abstract}
    We present explicit oracles designed to be used in Grover's algorithm to match investor preferences. Specifically, the oracles select portfolios with returns and standard deviations exceeding and falling below certain thresholds, respectively. One potential use case for the oracles is selecting portfolios with the best Sharpe ratios. We have implemented these algorithms using quantum simulators.
\end{abstract}

\section{Introduction}
Consider that the efficient frontier of a portfolio universe is calculated and presented to an investor, who would like to choose portfolios from it. The efficient frontier would consist of two lists of numbers - one containing the expected returns and the other containing the corresponding standard deviations of the optimal portfolios. Our implementation employs quantum algorithms to enable selection from optimal portfolios with specified risks and returns.

The goal of portfolio optimization is to find the asset allocations, which result in optimal portfolios with maximum expected returns for given risk levels or minimal risk for a given level of expected return.  Portfolios could be derived using the mean-variance method \cite{Markowitz}, where variances of asset prices are used as a measure of portfolio risk. The objective constructed in this scheme yields a nondominated set, which is called ``the efficient frontier''. It allows the investors to choose among the optimal (efficient) portfolios based on their preferences. To identify portfolios with preference (utility function) maximizing risk-return pairs, we utilize two quantum algorithms - the Quantum Exponential Search (QES) algorithm \cite{QES} and the Grover Adaptive Search (GAS) algorithm \cite{GAS}. Both algorithms leverage the quantum search algorithm called Grover's algorithm \cite{grover1996fast} to locate desired items in a list of items. A general application of quantum computing in finance is discussed in \cite{albareti}.

This paper shows how to construct oracles (see Section \ref{subsec:Grover}) to select items in a list, based on their values and applies the resulting algorithms to a financial use case. Here, the items are varying portfolio allocations in a given asset universe with their corresponding risks and returns.
In Section \ref{sec:lit_rev}, we name the sources, which are related to our work. Although this study concentrates on the implementation of a conditional search algorithm in the context of investment portfolios, we also outline the literature regarding quantum portfolio optimization. A background on the relevant quantum algorithms is given in Section \ref{sec:method}. Our hypothesis is presented in Section \ref{sec:Implementation}, along with the discussion of the implemented algorithms. Lastly, our experimental results are presented in Section \ref{sec:results}.

\section{Literature Review}\label{sec:lit_rev}

Grover's algorithm is based on amplitude amplification \cite{Brassard,Grover_1998}, where the probability of measuring the desired quantum states is amplified. Using this, QES returns one of the desired states after an expected number of $\mathcal O(\sqrt {N/M})$ calls to the oracle, where $M$ and $N$ are the number of desired states and the total number of possible states, respectively. Another example of Grover-based algorithms is the GAS algorithms. These are optimization algorithms, which utilize Grover search in an iterative scheme to find the optimum value of an objective function. The relevant GAS algorithm that we use in our work is \cite{GAS}. Including QES as a subroutine, GAS finds the minimum among the items in an unsorted table of size $N$, also after $\mathcal{O}(\sqrt N)$ calls to the oracle. An extension of GAS to the $k$-minima problem is investigated in \cite{k_minima}. The algorithms QES and GAS allow us to select portfolios from a list of optimal portfolios. We would like to stress that this approach differs from the following literature, where the calculation of efficient portfolios is studied.

An instance of GAS in the context of portfolio optimization is \cite{Gilliam_2021}, where oracles for Constrained Polynomial Binary Optimization (CPBO) problems using GAS are constructed and tested on IBM's gate-based hardware. A subclass of CPBO are Quadratic Unconstrained Binary Optimization (QUBO) problems, which the portfolio optimization problem falls into.
Since QUBO is the natural input of quantum annealers, a significant portion of quantum portfolio optimization implementations use annealers \cite{gomes2022empirical,lopez2015multi,palmer2021quantum,pop_dwave,Venturelli_2019}.

In \cite{lim2022quantum} a quantum version of the portfolio optimization algorithm by \cite{helmbold1998line} is given. Their quantum algorithm achieves a quadratic speed-up in the time complexity with respect to the number of assets in the portfolio due to quantum state preparation and norm estimation with the assumptions of no short selling and quantum query access to asset returns. The latter describes having access to the desired state vectors by means of a certain quantum operation. Experimental results with 14 assets from S\&P500 using Honeywell's trapped-ion System Model H1 are shown in \cite{yalovetzky2021nisq}. Their algorithm is based on the hybrid HHL algorithm \cite{lee2019hybrid}, where the phase estimation is enhanced by mid-circuit measurement, quantum conditional logic, and qubit reset and reuse. In a follow-up paper, the authors carry out the constrained portfolio optimization on the same hardware, where they utilize the Quantum Zeno Effect in Quantum Approximate Optimization Algorithm and Layer Variational Quantum Eigensolver. By means of repeated projective measurements leading to Zeno dynamics, they achieve a higher in-constraint probability compared to the state-of-the-art technique of enforcing constraints by introducing a penalty into the objective. Details on some of the theoretical results regarding quantum portfolio optimization \cite{SOCP,HHL_port_opt} can be found in \cite{albareti}.
\section{Preliminaries}\label{sec:method}

In this section, the relevant concepts are introduced and the notation is set.

\subsection{Grover's search algorithm}\label{subsec:Grover}

Given an $N$ item unstructured search problem with $M$ solutions, Grover's algorithm finds the solutions with $\mathcal{O}(\sqrt N)$ calls to the oracle \cite{nielsenchuang}, where $M \in \mathbb{N}$ and for the sake of convenience, we set $N=2^n$ for an $n \in \mathbb{N}^+$. We also assume that $M \leq N/2$, since this can always be ensured by introducing an additional qubit to the system that doubles the size of the search space. At each Grover iteration, the initial state is rotated by an angle $\theta$ towards the solution state $\ket{\beta}$ by means of the Grover operator 
\begin{equation}\label{grover_op}
    G = \begin{pmatrix}
        \cos\theta & -\sin\theta\\
        \sin\theta & \cos\theta
    \end{pmatrix}
\end{equation}
on the plane spanned by the basis states $\ket{\beta}$ and $\ket{\alpha} \coloneqq \ket{\beta}^\perp$. Here, the angle between $\ket{\psi} \coloneqq \ket{+}^n$ and $\ket{\alpha}$ is given by 
\begin{equation}\label{theta}
    \frac{\theta}{2} = \arcsin\sqrt{\frac{M}{N}}.
\end{equation}
\begin{figure}[!htbp]
     \centering
     \includegraphics[scale=0.15]{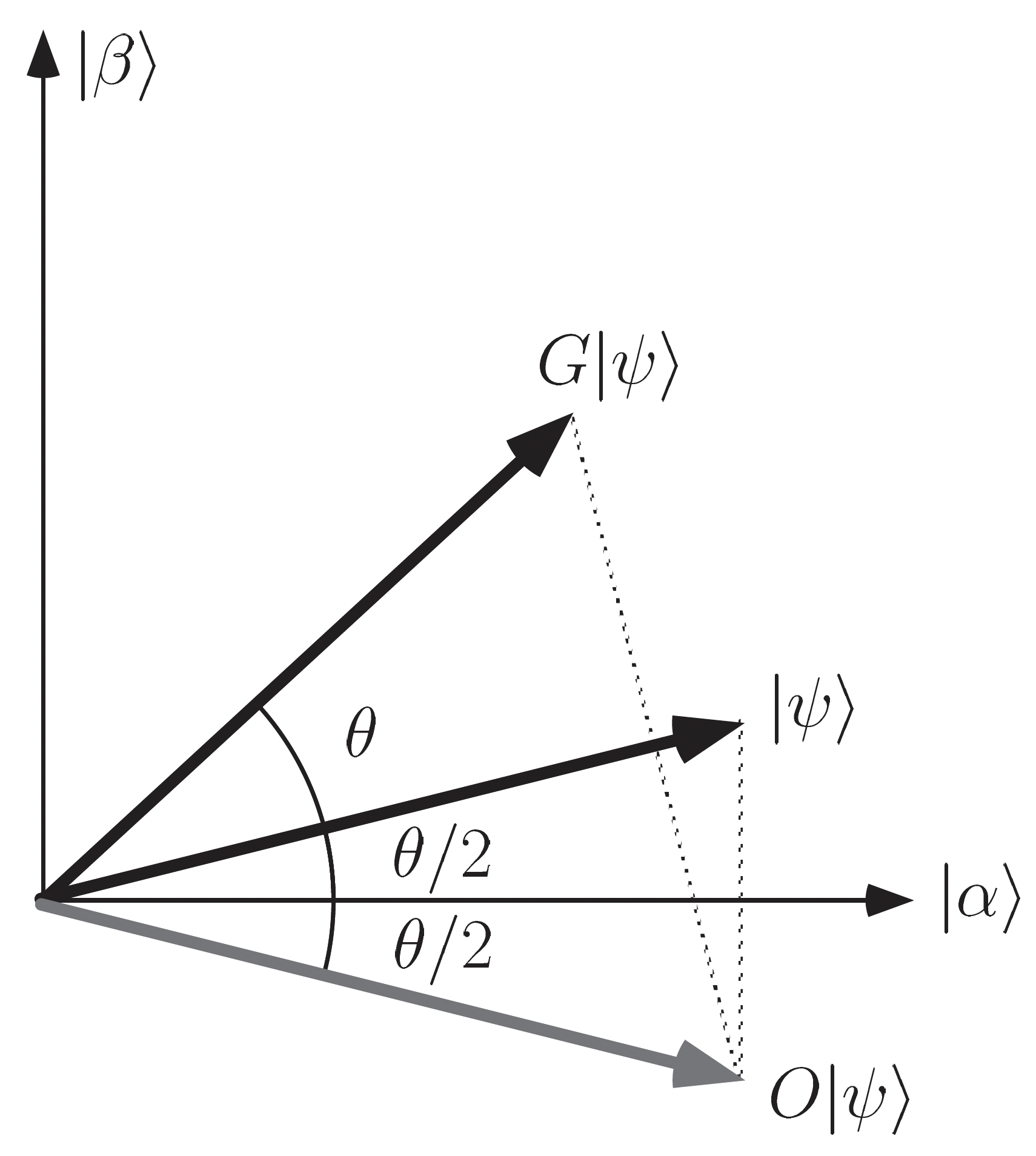}
     \caption{Geometrical illustration of the first Grover iteration. Source of image: \cite{nielsenchuang}.}
     \label{fig:grover}
\end{figure}

The quantum search algorithm is formulated in terms of an \emph{oracle}. An oracle is a `black box', that is able to `recognize' the solutions of a given problem. Upon input, the oracle $f$ outputs $1$ if the input is a solution to the problem we want to solve and $0$ otherwise. In quantum computing, an oracle is implemented as a unitary operator $O$ acting on the computational basis as
\begin{equation}\label{oracle}
    \ket{x}\ket{q} \xrightarrow{O} \ket{x}\ket{q\oplus f(x)},
\end{equation}
where $\oplus$ denotes addition modulo 2, $\ket{q}$ is the oracle qubit constituting the decision of the oracle and the input is registered in $\ket{x}$. If $\ket{q}$ is prepared in the $\ket{-}$ state, \eqref{oracle} becomes
\begin{equation*}%phase_kickback
    \ket{x}\ket{-} \xrightarrow{O} (-1)^{f(x)} \ket{x}\ket{-},
\end{equation*}
where the solutions are marked by a phase kickback. This corresponds to reflecting the initial state about $\ket{\alpha}$. The complete Grover rotation \eqref{grover_op} can be achieved by applying a second reflection about $\ket{\psi}$ (see Figure \ref{fig:grover}).
The number of Grover rotations required to reach $\beta$ within an angle $\theta/2 \leq \pi/4$ is given by
\begin{equation*}
    \textnormal{CI}\left(\frac{\pi/2}{\theta} - \frac{1}{2}\right),
\end{equation*}
where $\textnormal{CI}(x)$ denotes the integer closest to the real number $x$. By convention, halves are rounded down. Measurement of the state then yields a solution to the search problem with a probability of at least
one-half.

\subsection{Quantum phase estimation}\label{subsec:QPE}
If the Grover search involves non-integer values, the decimal parts need to be encoded in the quantum circuits. One way to achieve this is to use quantum phase estimation. It estimates the strength $\varphi$ of a quantum phase $e^{i2\pi\varphi}$ in units of $2\pi$ to a specified accuracy, where $\varphi \in [0,1)$ and $e^{i2\pi\varphi}$ is an eigenvalue of a unitary operator $U$. In order to have an accuracy of $2^{-m}$ and a success probability of $1-\epsilon$, the number of phase estimating qubits must be
\begin{equation*}
    t = m + \log (2 + 1/2\epsilon).
\end{equation*}
The estimated phase strength $\widetilde\varphi$ is then represented as an integer $b$ in $t$ qubits following the relation
\begin{equation}\label{phase_est_integer_formula}
    b = 2^t\widetilde\varphi.
\end{equation}
The runtime of quantum phase estimation is $\mathcal{O}(t^2)$ \cite{nielsenchuang}.

\subsection{Quantum counting}\label{Approximate}

Grover's algorithm serves to find the solutions of a search problem with fewer consultations to the oracle. Classically, we would need $\mathcal{O}(N)$ consultations to estimate $M$. Utilizing Grover operators and quantum phase estimation, quantum counting can speed this process up. Moreover, it can be used to estimate whether there is a solution to the problem or not. 

The number of solutions $M$ and the Grover angle $\theta$ are related by \eqref{theta}. Hence, choosing the unitary operator $U$ mentioned in Section \ref{subsec:QPE} as the Grover operator \eqref{grover_op} with the eigenvalues $e^{\pm i\theta}$, quantum phase estimation yields an estimate of $M$. The error related to the estimation of $M$ is given by \cite{nielsenchuang}
\begin{equation}\label{error_M}
    \Delta M < 2^{-m} \left ( \sqrt{NM} + \frac{N}{4} \cdot 2^{-m} \right ).
\end{equation}

\subsection{Quantum Exponential Search}\label{QES}
QES is a generalization of Grover's algorithm, allowing us to use Grover's algorithm without prior knowledge of the number of solutions. After an expected number of $\mathcal O(\sqrt {N/M})$ calls to the oracle, a solution is returned almost certainly. Basically, it employs Grover's algorithm in a loop, where each time the upper bound of the range is increased, which the number of Grover iterations is uniformly randomly chosen from. The program is exited upon measuring a solution and the case of no solution is handled by an appropriate time-out. We state QES in Algorithm \ref{algo:QES} for the reader's convenience.
\begin{algorithm}
    \caption{QES}
    \begin{algorithmic}[1]
    \REPEAT
        \STATE $m \gets 1$
        \STATE $\lambda \gets 8/7$
        \STATE Choose integer $j$ uniformly at random with $0\leq j < m$
        \STATE Apply $j$ Grover iterations
        \STATE Measure $i$
         \IF{$f(i) = 1$}
            \STATE Exit
        \ELSE
            \STATE $m \gets$ min$(\lambda m , \sqrt{N})$
        \ENDIF
    \UNTIL{Program is exited}
    \end{algorithmic}
    \label{algo:QES}
\end{algorithm}

\subsection{Grover Adaptive Search}\label{GAS}
As mentioned in Section \ref{sec:lit_rev}, GAS is a Grover-based algorithm to find the minimum value in a list by iteratively applying Grover search and using the best-known value as a threshold to flag all values smaller than the threshold. It uses QES as a subroutine and returns the index of the minimum with a success probability of at least one-half after making $\mathcal{O}(\sqrt{N})$ calls to the oracle.  We state GAS in Algorithm \ref{algo:GAS} for the reader's convenience.
\begin{algorithm}
    \caption{GAS}
    \begin{algorithmic}[1]
    \STATE Choose integer $j$ uniformly at random with $0\leq j < N$
    \REPEAT
        \STATE Initialize memory as $\sum_i \frac{1}{\sqrt{N}}\ket{i}\ket{j}$
        \STATE Mark items whose values are smaller than the $j^{th}$ value
        \STATE Apply QES with outcome $j'$
        \STATE $j\gets j'$
    \UNTIL{Number of calls to the oracle has exceeded $22.5\sqrt{N}+1.4\log^2(N)$}
    \end{algorithmic}
    \label{algo:GAS}
\end{algorithm}
\section{Implementation}\label{sec:Implementation}

The research question initially aimed to investigate the feasibility of implementing portfolio optimization with continuous asset allocations and positivity constraints on gate-based quantum hardware or simulators. However, to the best of our knowledge, the only available literature on this topic proposes an algorithm that primarily focuses on theoretical runtime rather than practical implementation \cite{SOCP}. After an analysis that demonstrates the infeasibility of implementing the complete portfolio optimization problem with these criteria on existing quantum hardware or simulators, the research objective has been redefined to focus on selecting portfolios from the set of efficient portfolios that maximize investor preferences. In this context, the hypothesis posits that by utilizing Grover-based algorithms with suitable oracles, it is possible to identify portfolios that maximize investor preferences, assuming the risk-return pairs of the optimal portfolios are available as input. In the following, we assume that the investor preferences are maximized by selecting portfolios with returns and
standard deviations higher and lower than the desired threshold
values, respectively.

In this section, the gates and circuits required to select the preferred portfolios are presented. Once the efficient portfolios are generated by classical mean-variance optimization with no-short-selling condition, we select portfolios with returns and standard deviations higher and lower than desired threshold values, respectively. This is a conditional slicing problem with two lists and conditions. It is solved by employing QES with the oracles given in Section \ref{subsubsec:oracle}. The second part of our implementation deals with searching for the portfolio with the maximum Sharpe ratio, which is described in Section \ref{subsec:max_min}. Lastly, the scaling of the number of required qubits is analyzed in Section \ref{subsec:tot_no_qubits}.

\subsection{Conditional Slicing}
We explore the quantum version of selecting list items by a given condition, also known as (list) slicing by condition. Specifically, we are interested in conditions involving comparisons between non-negative values.

\subsubsection{GT-gate}\label{appendix:greater_than_gate}
In order to carry out the comparisons between values, a greater-than operation (GT-gate) is necessary. 
Our implementation of a GT-gate for the comparison $a>b$ is given by
\begin{equation}\label{GT-gate}
\begin{split}
    CX&(1,n+1)\cdots CX(n,2n)\\
    &MCX(n,2n,2n+1)X(2n) \\
    &MCX(n-1,2n-1,2n,2n+1)X(2n-1)\cdots \\
    &MCX(1,n+1,\dots,2n,2n+1)X(n+1)\\
    &X(n+1)\cdots X(2n)\\
    &CX(1,n+1)\cdots CX(n,2n)\ket{0}^{2n+1},
\end{split}
\end{equation}
where $a,b \in \mathbb{N}$. The last qubit carries the outcome of the comparison. The non-negative integers $a$ and $b$ are encoded bitwise in the first $n$ qubits and the remaining qubits, respectively. Note that the leading order bits have higher indices, i.e. they are located at lower positions in Circuit 1. The last two lines of \eqref{GT-gate} are for uncomputation.
\begin{figure}[!htbp]
     \centering
     \includegraphics[scale=0.28]{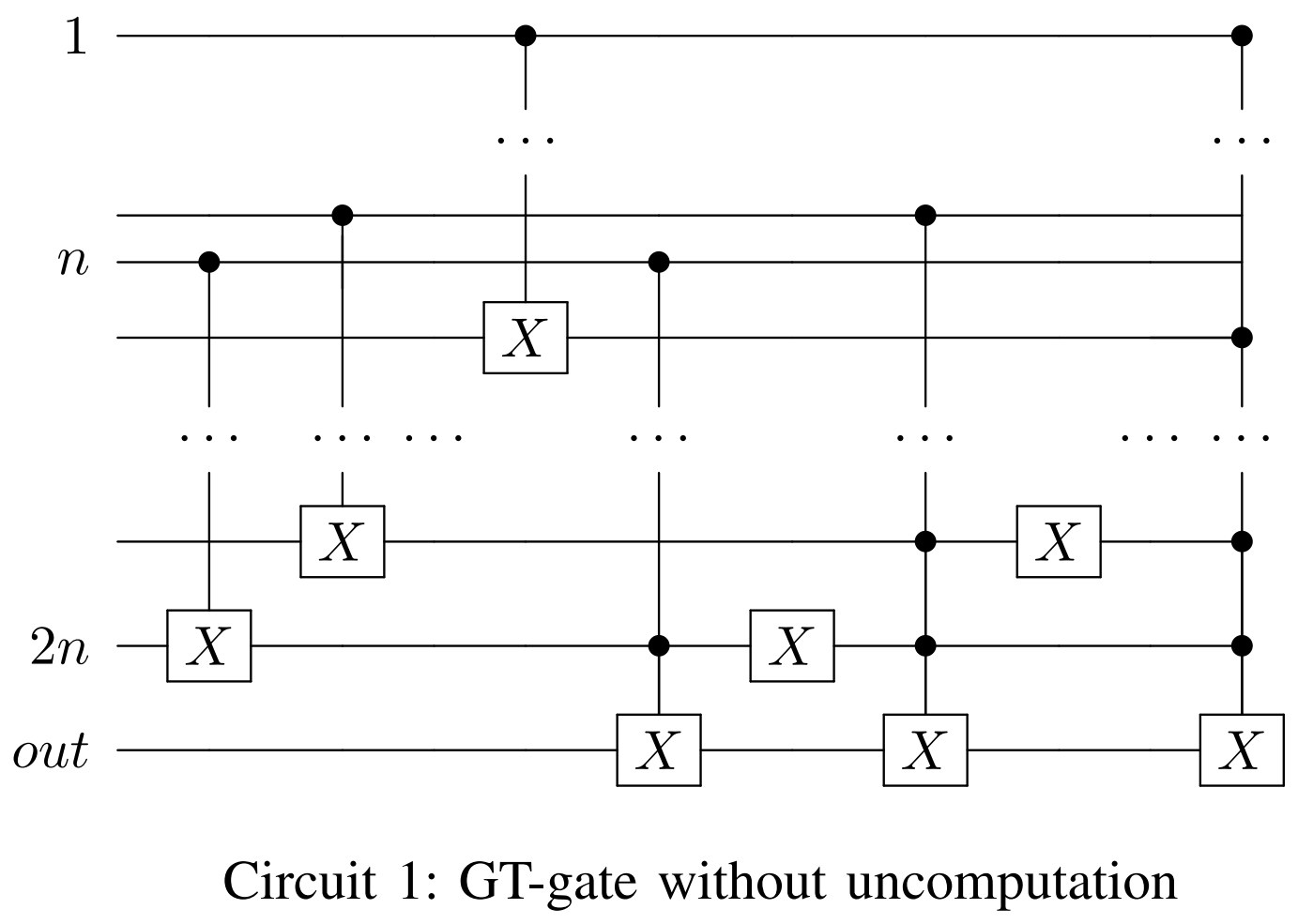}
\end{figure}
The circuit consists of NOT ($X$), controlled-NOT ($CX$) and multi-controlled-NOT ($MCX$) gates.
The number of all gates scales as $\mathcal{O}(n)$.

\subsubsection{Constructing Oracles}\label{subsubsec:oracle}

The quantum list slicing algorithm is based on Grover search. For a list with items $s_k$ and a threshold value $S$, the oracle must check each item for $s_k > S$, $k\in \{1,\dots,N\}$. In the case of non-integer values, the GT-gate must be applied to the integer part and the decimal part, separately. Since the comparison of integers with GT-gate is straightforward, we first concentrate on the decimal parts comparison. Assuming that we have access to the decimal parts in the form of 
\begin{equation}\label{U}
    U = diag(e^{i2\pi s_1},\dots,e^{i2\pi s_N}),
\end{equation}
they can be phase-encoded with the associated unitary operator $U$, resulting in integer representations of their $t$-bit approximations in the register, given by \eqref{phase_est_integer_formula}. This allows us to apply GT-gate on them.
In a similar manner, GT-gate can be applied to the integer parts by having access to their phase representations in the form of $\eqref{U}$.
For simplicity, we assume that all values are in the range $[0,1)$. A diagram of the initial state and the oracle is shown below:
\[
\begin{array}{c}
\Qcircuit @C=1.2em @R=1em {
\lstick{\ket{+}^n} & \gate{PE(U)} & \multigate{2}{GT} & \qw     &       &    \\
\lstick{\ket{S}}    & \qw           & \ghost{GT}        & \qw   &\cdots & \rotatebox{270}{Uncomputation} \\
\lstick{\ket{-}}    & \qw           & \ghost{GT}        & \qw   &       &    \\
} \\ \\
\textnormal{Circuit 2: Initial state and oracle. Single list.}
\end{array}
\]
$PE$ stands for phase estimation. The initial state (Grover eigenstate) is $\psi_0 = \ket{+}^n\ket{S}\ket{-}$.
This way, the first register carries the superposition of the list items. The second register has $S$. The comparison $s_k > S$ is realized by the GT-gate. Finally, the decision of the oracle is registered at the third register. 

Slicing by multiple conditions, e.g. 
\begin{equation*}
    S_2>s_k > S_1    
\end{equation*}
or
\begin{equation*}
    S_2>s_k \text{ AND } s_k = S_1,
\end{equation*}
can be done in an analogous way, where $S_1$ and $S_2$ are threshold values. An interesting case is when multiple conditions are based on different lists, such as
\begin{equation*}
    S_2>\sigma_k \text{ AND } r_k > S_1.
\end{equation*}
A diagram of the oracle for such a slicing is shown below:
\[
\begin{array}{c}
\Qcircuit @C=1.3em @R=0.5em {
\lstick{\ket{+}^n} & \ctrl{1} & \gate{PE(U_r)} & \multigate{2}{GT} & \qw & \qw & \\
\lstick{\ket{S_1}} & \qw & \qw & \ghost{GT} & \qw & \qw & \\
\lstick{\ket{0}} & \qw & \qw & \ghost{GT} & \ctrl{3} & \qw & \\
\lstick{\ket{S_2}} &\qw &  \qw & \multigate{2}{GT} & \qw & \qw & \cdots & \rotatebox{270}{Uncomputation}\\
\lstick{\ket{0}^n} & \sgate{X^n}{-3} & \gate{PE(U_\sigma)} & \ghost{GT} & \qw & \qw & \\
\lstick{\ket{0}} & \qw & \qw & \ghost{GT} & \ctrl{1} & \qw & \\
\lstick{\ket{-}} & \qw & \qw & \qw & \gate{X} & \qw & 
} \\ \\
\textnormal{Circuit 3: Initial state and oracle for two lists.}
\end{array}
\]
The unitary operators $U_r$ and $U_\sigma$ contain the list items $r_k$ and $\sigma_k$ in the previously given form \eqref{U}, respectively. 

A point to elaborate on is the resolving power of the GT-gate, when comparing phase-encoded values. The resolving power quantifies how
close two values can be, while the GT-gate can
successfully compare them. If two distinct values $\varphi_1,\varphi_2 \in [0,1)$ are represented by integers $b_1,b_2 \in \mathbb{N}$ according to \eqref{phase_est_integer_formula}, we would like to satisfy the separation condition $|b_1 - b_2| \geq 1$, which implies
\begin{equation}\label{gt_separation_rule}
    t \geq \log |\varphi_1 - \varphi_2|^{-1}.
\end{equation}
For a desired resolving power $ d$, \eqref{gt_separation_rule} implies that we can choose
\begin{equation}\label{t_G}
    t = \ceil[\bigg]{\log  d^{-1}}.
\end{equation}

\subsubsection{The Algorithm}\label{subsubsec:algorithm1}

The quantum conditional slicing is realized by inserting the initial states and the oracles presented in Section \ref{subsubsec:oracle} into QES. One of the solutions is returned with equal probability after an expected number of $\mathcal O(\sqrt {N/M})$ calls to the oracle. In case there are multiple solutions, one needs to run the algorithm multiple times to obtain them. By Theorem \ref{theorem2}, it can be checked whether there is no solution, a single solution or multiple solutions with quantum counting after $\mathcal{O}(\sqrt{N})$ calls to the oracle by choosing
\begin{equation}\label{m}
        m = \ceil[\bigg]{\frac{\log N}{2} + 1.583}
    \end{equation} 
and rounding the measured number of solutions to the closest integer. According to Theorem \ref{theorem1}, the exact number of solutions can be obtained after $\mathcal{O}(N)$ calls to the oracle by choosing
\begin{equation}\label{m_other}
    m = \ceil[\bigg]{\log N + \frac{1}{2}}.
\end{equation}
Then, the regular Grover's algorithm with the correct number of Grover iterations can also be used instead of QES, since the number of solutions is known.

\subsection{Finding the maximum}\label{subsec:max_min}

We can use the oracle in Circuit 2 from Section \ref{subsubsec:oracle} in GAS to get the maximum value in a list. At the first iteration, the threshold value in the initial state is prepared as the list value with the index that is chosen uniformly at random. At the later iterations, it is substituted by the outcome of QES. Implementation of a minimum function is analogous. After running GAS for $c$ times, the probability of success is at least $1-1/2^c$.

\subsection{Total number of qubits}\label{subsec:tot_no_qubits}
The number of qubits required to encode the list elements scales logarithmically with the number of elements in the list. Additionally, the choices \eqref{m} and \eqref{m_other} for $m$ show that the number of counting qubits scales as $\mathcal{O}(\log N)$. From \eqref{t_G} we see that the number of phase estimating qubits is independent of $N$, but dependent on the resolving power $ d$. Hence, the total number of qubits required by both algorithms scales as $\mathcal{O}\left(\log (N { d}^{-1})\right)$.
\section{Results}\label{sec:results}
    The circuits that are described in Section \ref{sec:Implementation} are implemented using the simulators from IBM's Qiskit SDK \cite{Qiskit}.
    Our results are accessible at \cite{QAIFgit}, where the corresponding repository includes other gates, such as Less Than, Equals, and OR gate, which enable different selection conditions. 
    
    The first part of our implementation shows that the desired portfolios with specified risks and returns are found by employing our oracles in QES. Note that we take here a subsample of the optimal portfolios resulting in eight values for returns and standard deviations each. 
 This way we reduce the number of qubits required to hold the values, which is too high in terms of the required memory, otherwise.
 In the second part of our implementation, Sharpe ratios of the portfolios are calculated, where a risk-free rate of $0\%$ is chosen for simplicity. Then, the index of the portfolio with the highest Sharpe ratio is successfully obtained by using Circuit 2 in GAS. Note that all values in both experiments are in $[0,1)$ and we do not need to compare integer parts. A resolving power $d=0.01$ is chosen for both experiments, resulting in fifteen qubits for GT-gate. Three qubits are required for encoding the indices of eight values. Hence, Circuit 2 and 3 have eighteen and thirty-seven qubits, respectively. For the case $M=0$, a time-out of $\mathcal{O}(\sqrt{N})$ is set in QES. Quantum counting with \eqref{m} can be used as a substitute to the handling of the case $M=0$ with a time-out. This could be implemented as a termination condition in GAS.

\section{Conclusion}
 The described financial use case for quantum computing regarding portfolio selection by investor preferences has been successfully implemented by employing the Grover-based algorithms QES and GAS with the oracles, whose explicit forms are given in this paper. To advance this research, one can implement the proposed method on physical quantum hardware to assess its performance in the presence of noise. This approach would also enable the inclusion of larger portfolios by leveraging devices with a higher number of qubits than those currently available in simulators. Nevertheless, the small-scale examples in \cite{QAIFgit} are conceptually viable.
\appendix
\subsection{Theorems} 

\begin{theorem}\label{theorem1}
    The error $\Delta M$ of the quantum counting algorithm is less than $1/2$, if the list size $N$ is large, the number of solutions $M$ is at most $N/2$ and the number of bit accuracy qubits is chosen as 
    \begin{equation*}
        m = \ceil[\bigg]{\log N + \frac{1}{2}}.
    \end{equation*} 
\end{theorem}

\begin{proof}
    Suppose we want to bound $\Delta M$ by $\varepsilon$, i.e. $\Delta M < \varepsilon$. Then, \eqref{error_M} implies
\begin{equation}\label{deriv_2}
    2^{-m}\sqrt{MN} + \frac{N}{4}2^{-2m} \leq \varepsilon.   
\end{equation}
After defining
\begin{equation}\label{deriv_f}
    f(m,N) \coloneqq 2^{-m}\sqrt{N},
\end{equation}
from \eqref{deriv_2} we get
\begin{alignat}{2}
    &&2^{-m}\sqrt{MN} + \frac{N}{4}2^{-2m} \leq \varepsilon \nonumber \\
    &\Leftrightarrow & 4\sqrt{M} f + f^2 \leq 4\varepsilon \nonumber \\
    &\Leftrightarrow &  (f + 2 \sqrt{M})^2 \leq 4(\varepsilon + M) \nonumber \\
    &\Rightarrow &      f + 2 \sqrt{M} \leq 2\sqrt{\varepsilon + M} \nonumber \\
    &\Rightarrow & 2^{m+1} \geq \sqrt{N} \frac{\sqrt{M + \varepsilon} + \sqrt{M} }{\varepsilon}\label{m_M_rel1}
\end{alignat}
where in the last line we have used
\begin{equation*}
    2\sqrt{\varepsilon + M} - 2\sqrt{M} = \frac{2\varepsilon}{\sqrt{M+\varepsilon} + \sqrt{M}}
\end{equation*}
and plugged \eqref{deriv_f} back in.
Plugging $\varepsilon = 1/2$ into \eqref{m_M_rel1} gives
\begin{alignat}{2}
    && 2^{m+1} \geq \sqrt{2N}(\sqrt{2M+1} + \sqrt{2M}) \nonumber\\
    &\Rightarrow& 2^{m+1} \geq \sqrt{2N}(\sqrt{N+1} + \sqrt{N}) \nonumber\\
    &\Leftrightarrow& 2^{m+1} \geq \frac{\sqrt{2}}{\sqrt{1+1/N}-1} \nonumber\\
    &\Rightarrow& m \geq - \frac{1}{2} - \log(\sqrt{1+1/N}-1),\label{deriv_3}
\end{alignat}
where in the second line we have used $M\leq N/2$.
Assuming $N\gg1$, we have
\begin{equation}
    \log(\sqrt{1+1/N}-1) \approx \log(1/2N).\label{deriv_4}
\end{equation}
Combining \eqref{deriv_3} with \eqref{deriv_4} gives 
\begin{equation*}
    m\geq\log N + \frac{1}{2},
\end{equation*}
which is the desired result.
\end{proof}

\begin{theorem}\label{theorem2}
   Quantum counting algorithm checks whether there is no solution, a single solution, or multiple solutions if the number of bit accuracy qubits is chosen as 
    \begin{equation*}
        m = \ceil[\bigg]{\frac{\log N}{2} + 1.583}.
    \end{equation*} 
\end{theorem}

\begin{proof}
   Choosing
   \begin{equation*}
       \varepsilon=M-3/2
   \end{equation*}
    for all $M \in \{2,3,\dots,N/2\}$ and plugging $\varepsilon$ into \eqref{m_M_rel1} gives
\begin{equation*}
    2^{m+1} \geq \sqrt{N}\frac{\sqrt{2M-3/2}+\sqrt{M}}{M - 3/2}.
\end{equation*}
The right-hand side is maximized at $M=2$ yielding
\begin{alignat*}{2}
    && 2^{m+1} \geq \sqrt{N}\frac{\sqrt{5/2}+\sqrt{2}}{1/2} \\
    &\Rightarrow& m \geq - \frac{1}{2} + \frac{1}{2} \log N + \log(2+\sqrt{5}).
\end{alignat*}
Since $\log(2+\sqrt{5})-1/2 \approx 1.583$, we can choose
    \begin{equation}\label{m_deriv}
        m = \ceil[\bigg]{\frac{\log N}{2} + 1.583}.
    \end{equation}
    For the cases $M \in \{0,1\}$, we need to choose $\varepsilon=1/2$, but those cases are already covered by \eqref{m_deriv}.
\end{proof}
\printbibliography

\end{document}